\def\year{2024}\relax
\documentclass[letterpaper]{article} % DO NOT CHANGE THIS
\usepackage{comment}
\usepackage{aaai24}  % DO NOT CHANGE THIS
\usepackage{times}  % DO NOT CHANGE THIS
\usepackage{helvet}  % DO NOT CHANGE THIS
\usepackage{courier}  % DO NOT CHANGE THIS
\usepackage[hyphens]{url}  % DO NOT CHANGE THIS
\usepackage{graphicx} % DO NOT CHANGE THIS
\usepackage{multirow}
\usepackage{booktabs}
\usepackage{algorithmic}
\usepackage{bbding}
\usepackage{algorithm}
\usepackage[reqno]{amsmath}
\usepackage{amssymb}
\usepackage{tikz}
\usepackage{graphicx}
\usepackage{subcaption}
\newtheorem{fact}{Fact}
\newtheorem{definition}{Definition}
\newtheorem{lemma}{Lemma}
\newtheorem{theorem}{Theorem}

\newenvironment{proof}{\paragraph{Proof:}}{\hfill$\square$}
\usepackage{natbib}  % DO NOT CHANGE THIS AND DO NOT ADD ANY OPTIONS TO IT
\usepackage{caption} % DO NOT CHANGE THIS AND DO NOT ADD ANY OPTIONS TO IT

\DeclareCaptionStyle{ruled}{labelfont=normalfont,labelsep=colon,strut=off} % DO NOT CHANGE THIS
\usepackage{algorithm}
\usepackage{algorithmic}

\usepackage{newfloat}
\usepackage{listings}
\floatstyle{ruled}
\newfloat{listing}{tb}{lst}{}
\floatname{listing}{Listing}
\title{A Generalized Shuffle Framework for Privacy Amplification: Strengthening Privacy Guarantees and Enhancing Utility}
\author {
    E Chen\textsuperscript{\rm 1},
    Yang Cao\textsuperscript{\rm 2, \thanks{Corresponding author: Yang Cao, yang@ist.hokudai.ac.jp }},
    Yifei Ge \textsuperscript{\rm 3}
}
\affiliations {
    \textsuperscript{\rm 1} Zhejiang Lab\\
    \textsuperscript{\rm 2} Hokkaido University\\
     \textsuperscript{\rm 3} Xi’an Jiaotong-Liverpool University \\
    chene@zhejianglab.com, yang@ist.hokudai.ac.jp, 
    Yifei.Ge23@student.xjtlu.edu.cn
}

\usepackage{bibentry}
\begin{document}
\maketitle

\begin{abstract}
The shuffle model of local differential privacy is an advanced method of privacy amplification designed to enhance privacy protection with high utility. 
It achieves this by randomly shuffling sensitive data, making linking individual data points to specific individuals more challenging.
However, most existing studies have focused on the shuffle model based on
$(\epsilon_0,0)$-Locally Differentially Private (LDP) randomizers, with limited consideration for complex scenarios such as $(\epsilon_0,\delta_0)$-LDP or personalized LDP (PLDP). 
This hinders a comprehensive understanding of the shuffle model's potential and limits its application in various settings.
To bridge this research gap, we propose a generalized shuffle framework that can be applied to any
$(\epsilon_i,\delta_i)$-PLDP setting with personalized privacy parameters. 
This generalization allows for a broader exploration of the privacy-utility trade-off and facilitates the design of privacy-preserving analyses in diverse contexts.
We prove that shuffled $(\epsilon_i,\delta_i)$-PLDP process approximately preserves $\mu$-Gaussian Differential Privacy with $
\mu = \sqrt{\frac{2}{\sum_{i=1}^{n} \frac{1-\delta_i}{1+e^{\epsilon_i}}-\max_{i}{\frac{1-\delta_{i}}{1+e^{\epsilon_{i}}}}}}.
$
This approach allows us to avoid the limitations and potential inaccuracies associated with inequality estimations.
To strengthen the privacy guarantee, we improve the lower bound by utilizing \textit{hypothesis testing} instead of relying on rough estimations like the Chernoff bound or Hoeffding's inequality.
Furthermore, extensive comparative evaluations clearly show that our approach outperforms existing methods in achieving strong central privacy guarantees while preserving the utility of the global model.
We have also carefully designed corresponding algorithms for average function, frequency estimation, and stochastic gradient descent.
\end{abstract}

\section*{Introduction}
The shuffle model \cite{Bittau2017prochlo} is a state-of-the-art technique to balance privacy and utility for differentially private data analysis.
In traditional differential privacy, a trusted server (or aggregator) is often assumed to collect all users' data before privacy-preserving data analysis \cite{Dwork2014algorithmic}. 
However, such approaches may not be feasible or practical in scenarios where a trusted curator does not exist.  
Given this, Local Differential Privacy (LDP) \cite{KS11} has been proposed to achieve differential privacy by allowing the users to add noises individually; however, LDP suffers from low utility due to the accumulated noise. %This is where the shuffle model comes into play.
To address this, the shuffle model of differential privacy (shuffle DP) \cite{Bittau2017prochlo, balle2019privacy, Erlingsson2019amplification} adds a shuffler between the users and the server to randomly shuffle the noisy data before sending the server.
The shuffle DP has an intriguing theoretical privacy amplification effect, which means a small amount of local noise could result in a strong privacy guarantee against the untrusted server.
%By employing the shuffle model, differential privacy techniques can provide stronger privacy guarantees by adding an extra layer of protection. 
%In the shuffle model, data records are randomly shuffled or rearranged before applying any privacy-preserving mechanism. 
%This process effectively breaks the direct link between the original data points and the individuals they belong to. 
%However, there is an urgent need for research addressing privacy amplification in the shuffle model under a more general framework.
% The concept of shuffle model was first introduced in the work of  privacy-preserving software monitoring \cite{Bittau2017prochlo}. 
Extensive studies \cite{balle2019privacy,Erlingsson2019amplification,Girgis2021renyi,Feldman2022hiding, liu2021flame,GDDSK21federated} have been devoted to proving a better (tighter) privacy amplification in the shuffle DP.

However, most existing studies have focused on the shuffle model based on $(\epsilon_0,\delta_0)$-LDP randomizer with uniform and limited settings of local privacy parameters $\epsilon_0$ and $\delta_0$.
For example,  \citealt{Erlingsson2019amplification} assumes $0<\epsilon_0<1/2$ and $\delta_0 = 0$.
%Moreover, complex scenarios, such as $(\epsilon_0,\delta_0)$-LDP or personalized LDP, have received insufficient attention.  
Although a recent work \citealt{Liu2023} provides a privacy bound for local personalized privacy parameter $\epsilon_i$  for each user $i$ (and a fixed $\delta_0$), the bound is relatively rough and has a large room to be improved.
To address this problem, we make the following contributions.

Firstly, we propose a {\bf G}eneralized {\bf S}huffle framework for {\bf P}rivacy {\bf A}mplification (GSPA) to allow arbitrary local privacy parameters and provide new privacy amplification analysis.
Our analysis technique benefits from the adoption of Functional Differential Privacy \cite{Dong2022gaussian} and carefully analyzing the distance between two multinomial distributions (see Theorem \ref{thm:post3} and \ref{thm:GDPdistance}). 
For both uniform and personalized privacy parameter settings, we provide lower privacy bounds that exceed that of existing results (see Figure \ref{Fig:VsLiu}).
  
  % Our privacy bound (refer to Theorem \ref{thm:GDPdistance}) remains robust, showing minimal susceptibility to variations in individual users' privacy parameters, such as $\max_i{\epsilon_i}$ or $\max_{i,j}{\epsilon_i/\epsilon_j}$, which notably affected previous privacy bounds \cite{Liu2023}.
{Secondly, we apply GSPA with different personalized privacy parameter settings to diverse privacy-preserving analysis tasks, including private mean, private frequency estimation, and DP-SGD, to demonstrate the effectiveness of our approach. 
For mean and frequency estimation with GSPA (see Figure 3), the more conservative users there are, the less utility is observed, showing a negative linear relationship. Simultaneously, as the privacy parameters of conservative users increase, utility demonstrates a positive linear relationship.
For DP-SGD with GSPA (see Figure 4), there exists an interesting phenomenon that despite the constant scenario ($\epsilon_0=0.5$) offers a stronger privacy protection, its test accuracy (94.8\%) is higher than that (93.5\%) of scenarios $U(0.01,2)$, which have varying local privacy parameters.}
\section*{Preliminaries}
This section presents the definitions and tools necessary for understanding the shuffle model. These serve as fundamental tools for proposing our methods and form the basis of our approach.
\begin{definition}\label{ApprDP}({\bf Differential Privacy})
A randomized algorithm $\mathcal{R}$ %with domain $\mathbb{N}^{|\mathcal{X}|}$
satisfies $(\epsilon,\delta)$-differential privacy, denoted as $(\epsilon,\delta)$-DP, if for all $\mathcal{S} \subseteq$ Range$(\mathcal{R})$ and for all neighboring databases $D_0, D_1$ ($D_0$ can be obtained from $D_1$ by
replacing exactly one record):
\begin{equation}\label{epsilonDeltaDis}
\mathbb{P}( \mathcal{R}(D_0)\in \mathcal{S}) \leq e^\epsilon \mathbb{P}(\mathcal{R}(D_1) \in \mathcal{S})+\delta.
\end{equation}
\end{definition}
$\epsilon$ is known as the privacy budget, while
 $\delta$ is referred to as the indistinguishability parameter, which describes the probability of privacy leakage exceeding
$\epsilon$. Both $\epsilon$ and $\delta$ should be as small as possible, indicating stronger privacy protection.
\begin{definition}({\bf Local Differential Privacy})
A randomized algorithm $\mathcal{R}:\mathcal{D} \rightarrow \mathcal{S}$ satisfies $(\epsilon,\delta)$-local differential privacy, denoted as $(\epsilon,\delta)$-LDP, if for all pairs $x,x' \in \mathcal{D}$, $\mathcal{R}(x)$ and $\mathcal{R}(x')$ satisfies
\begin{equation}\label{epsilonDeltaDis}
\mathbb{P}( \mathcal{R}(x)\in \mathcal{S}) \leq e^\epsilon \mathbb{P}(\mathcal{R}(x') \in \mathcal{S})+\delta.
\end{equation}
\end{definition}
In Local Differential Privacy (LDP), each data contributor applies a local randomization mechanism to perturb their own data before sharing it with a central aggregator.
%The randomization introduces noise to the data, making it difficult to infer the exact contribution of any specific individual. By perturbing the data locally, LDP ensures that the privacy of each data contributor is protected even when the aggregated data is analyzed.

\begin{figure}
  \centering
  \includegraphics[scale = 0.35]{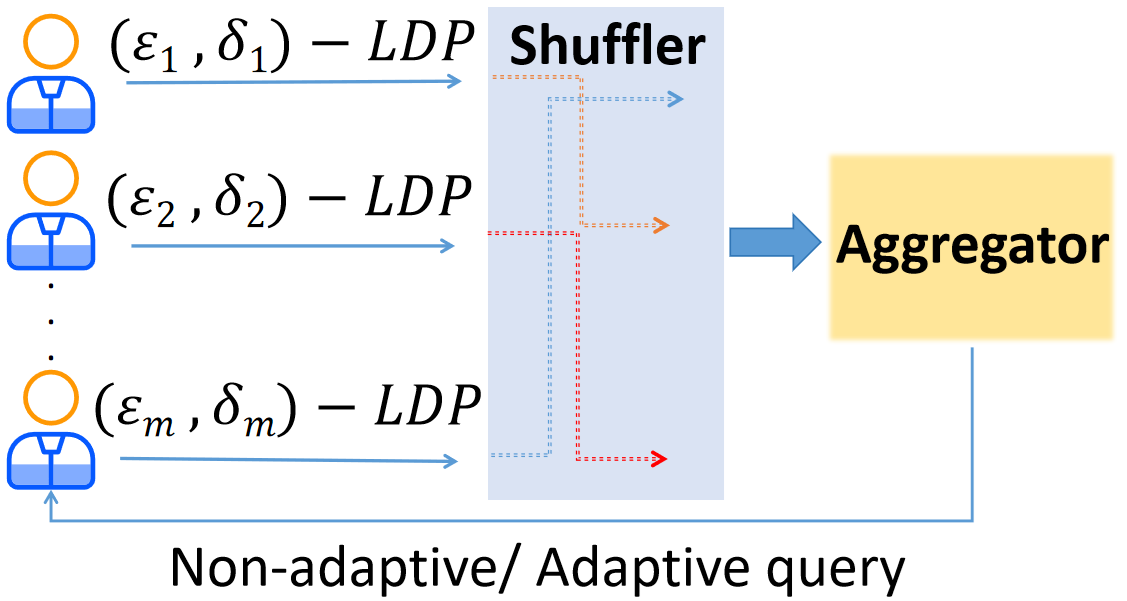}\\
  ~\\
  \caption{The {\bf G}eneralized {\bf S}huffle framework for {\bf P}rivacy {\bf A}mplification (GSPA).  Privacy parameters $(\epsilon_i,\delta_i)$ and each client's output  are shuffled separately. 
  The random permutation of the shuffler is unknown to anyone except the shuffler itself. 
  The type of query, whether non-adaptive or adaptive, depends on whether the next query depends on the previous output.
  %The expression of the aggregator is general, encompassing common tasks such as frequency estimation.
  }\label{Fig:ShufflePlot}
\end{figure}

\subsection{Privacy Tools}
Differential privacy can be regarded as a hypothesis testing problem for a given distribution \cite{Kairouz2015composition}.
 In brief, we consider the hypothesis testing issue with two hypotheses.
\begin{align*}
&H_0: \text{The underlying dataset is } D_0, \\
&H_1: \text{The underlying dataset is } D_1.
\end{align*}     
To provide an intuitive explanation, we designate the name Bob to denote the exclusive individual present in $D_0$ but absent in $D_1$.
Consequently, rejecting the null hypothesis implies the recognition of Bob's nonexistence, whereas accepting the null hypothesis suggests observing Bob's existence in the dataset.

Inspired by this, an effective tool called $f$-DP \cite{Dong2022gaussian} has been introduced, which utilizes hypothesis testing to handle differential privacy.
For two neighbouring databases $D_0$ and $D_1$, let $U$ and $V$ denote the probability distributions of $\mathcal{R}(D_0)$ and $\mathcal{R}(D_1)$, respectively. We consider a rejection rule $0 \leq \phi \leq 1$, with type I and type II error rates defined as
\begin{equation}\label{TypeError}
\alpha_\phi = \mathbb{E}_U[\phi], \quad \beta_{\phi} = 1-\mathbb{E}_V[\phi].
\end{equation}
It is well-known that
\begin{equation}\label{Tv2beta}
    \alpha_\phi+\beta_\phi \geq 1 - TV(U,V),
\end{equation}
where $TV(U,V)$ is the supremum of $|U(A)-V(A)|$ over all measurable sets $A$. To characterize the fine-grained trade-off between the two errors, Table 1 helps to establish a clear understanding of the relationship between the two errors.

\begin{table}
\setlength{\tabcolsep}{3pt} % 调整列距
\centering
\begin{tabular}{|c|c|c|}
\hline
 & Actual True & Actual False \\
\hline
Accept Hypothesis & Correct  & Type II Error ($\beta$) \\
\hline
Reject Hypothesis & Type I Error ($\alpha$) & Correct  \\
\hline
\end{tabular}
\label{Tab:TwoErrors}
\caption{Table of Error Types}
\end{table}

For any two probability distributions $U$ and $V$ on the same space $\Omega$, the trade-off function $T(U,V):[0,1]\rightarrow [0,1]$ is defined as
\begin{equation}
T(U,V)(\alpha) = \inf\{\beta_\phi: \alpha_\phi \leq \alpha \},
\end{equation}
where the infimum is taken over all measurable rejection rules $\phi$, and $\alpha_{\phi}=\mathbb{E}_U(\phi)$ and $\beta_{\phi}=1-\mathbb{E}_V(\phi)$.

\begin{definition} ({\bf Functional Differential Privacy}, $f$-DP)
Let $f$ be a trade-off function, a mechanism $\mathcal{R}$ is said to be $f$-DP if
\begin{equation}
T(\mathcal{R}(D_0),\mathcal{R}(D_1)) \ge f,
\end{equation}
for all neighboring data sets $D_0$ and $D_1$.
\end{definition}
To enhance readability, we have included the introduction and relevant properties of
$f$-DP in the section of Appendix. It is worth noting that traditional DP belongs to a special case of $f$-DP, therefore $f$-DP has a wider scope of applicability.

In addition, Laplace mechanism and Gaussian mechanism  are two common approaches used in differential privacy \cite{Dwork2014algorithmic}. The choice between the Laplace mechanism and the Gaussian mechanism depends on the data type, privacy requirements, and query tasks. The Laplace mechanism provides stronger privacy but may introduce larger errors, while the Gaussian mechanism is more suitable for accurate results. Thus, it's important to strike a balance between privacy and accuracy based on specific requirements.
\begin{definition}[Laplace Mechanism]
Given a query function $Q: D \rightarrow \mathbb{R}^d$, privacy parameter $\epsilon$ and $\ell_1$ sensitivity $\Delta (Q) = \max \| Q(D) - Q(D')\|_1$, then for any two neighbouring datasets $D,D'$,
the Laplace mechanism
\begin{equation}
M(D, Q) = Q(D) + \text{Lap}\left(\frac{\Delta (Q)}{\epsilon}\right)
\end{equation}
preserves $\epsilon$-DP, where \text{Lap}($\lambda$) denotes the centralized Laplace noise with scale parameter $\lambda$.
\end{definition}
In the absence of ambiguity, we express both queries and answers as $Q(\cdot)$ and $A(\cdot)$ respectively.
\begin{definition}[Gaussian Mechanism]
 Given a query function $Q: D \rightarrow \mathbb{R}^d$, privacy parameter $\epsilon$ and $\ell_2$ sensitivity $\Delta_2 (Q) = \max \|Q(D)-Q(D')\|_2$, then for any two neighbouring datasets $D,D'$, the Gaussian mechanism
\begin{equation}
M(D, Q) = Q(D) + N\left(0, \frac{2\log(1.25/\delta)\Delta_2^2(Q)}{\epsilon^2}\right)
\end{equation}
preserves $(\epsilon,\delta)$-DP,
where $N(\mu, \sigma^2)$ denotes the Gaussian noise with mean $\mu$ and variance $\sigma^2$.
\end{definition}

\section*{Privacy Analysis of GSPA Framework}
Our {\bf G}eneralized {\bf S}huffle framework for {\bf P}rivacy {\bf A}mplification (GSPA)
consists of three main components: local randomizers, a trustworthy shuffler, and an aggregator, which are the same as existing shuffle DP frameworks; however, GSPA allows local randomizers with arbitrary privacy parameters.
(i) For $n$ users, the local randomizer $M_i$ adds noise to the original data $x_i$ on the $i$-th user's devices, thus providing $(\epsilon^\ell_i,\delta^\ell_i)$-PLDP for user $i$. (ii) The shuffler randomly permutes the order of data elements, ensuring that the resulting arrangement is unknown to any party other than the shuffler itself. (iii) The aggregator collects and integrates shuffled data for simple queries, while for complex tasks like machine learning, it trains models based on shuffled data with multiple iterations.
Without causing confusion, the notation $(\epsilon_0, \delta_0)$-LDP is used to represent the uniform scenario, while $(\epsilon_i, \delta_i)$-PLDP denotes the personalized scenario.
\subsubsection*{Privacy Amplification Effect}
In this section, we address the issue of privacy protection in the context of a general shuffled adaptive process for personalized local randomizers.
\begin{definition}
For a domain $\mathcal{D}$, let $\mathcal{R}^{(i)}:\mathcal{S}^{(1)}\times \mathcal{S}^{(2)} \times \cdots \times \mathcal{S}^{(i-1)} \times \mathcal{D} \rightarrow \mathcal{S}^{(i)}$ for $i$ $\in [n]$, where $\mathcal{S}^{(i)}$ is the range space of $\mathcal{R}^{(i)}$ be a sequence of algorithms such that $\mathcal{R}^{(i)}(z_{1:i-1},\cdot)$ is an $(\epsilon_i,\delta_i)$-PLDP randomizer for all values of auxiliary inputs $z_{1:i-1} \in \mathcal{S}^{(1)} \times \mathcal{S}^{(2)} \times \cdots \mathcal{S}^{(i-1)}$. Let $\mathcal{A}_R: \mathcal{D} \rightarrow \mathcal{S}^{(1)} \times \mathcal{S}^{(2)} \times \cdots \times \mathcal{S}^{(n)}$ be the algorithm that given a dataset $x_{1:n} \in \mathcal{D}^n$, then sequentially computes $z_i = \mathcal{R}^{(i)}(z_{1:i-1},x_i)$ for $i \in [n]$ and outputs $z_{1:n}$. We say $\mathcal{A}_R(\mathcal{D})$ is a personalized LDP (PLDP) adaptive process. Similarly, if we first sample a permutation $\pi$ uniformly at random, then sequentially computes $z_i = \mathcal{R}^{(i)}(z_{1:i-1},x_{\pi_i})$ for $i \in [n]$ and outputs $z_{1:n}$, we say this process is shuffled PLDP adaptive and denote it by $\mathcal{A}_{R,S}(\mathcal{D})$.
\end{definition}

\begin{lemma}\label{R2Q}
Given an $(\epsilon_i,\delta_i)$-PLDP adaptive process, then in the $i$-th step, local randomizer  $\mathcal{R}^{(i)}$: $\mathcal{D} \rightarrow \mathcal{S}$ and for any $n+1$ inputs $x_1^0, x_1^1, x_2, \cdots, x_n \in \mathcal{D}$, there exists distributions
$\mathcal{Q}_1^0, \mathcal{Q}_1^1, \mathcal{Q}_1, \mathcal{Q}_2, \cdots, \mathcal{Q}_n$ such that
\begin{equation}
\mathcal{R}^{(i)}(x_1^0) = \frac{(1-\delta_i)e^{\epsilon_i}}{1+e^{\epsilon_i}}\mathcal{Q}_1^0+\frac{1-\delta_i}{1+e^{\epsilon_i}}\mathcal{Q}_1^1 + \delta_i \mathcal{Q}_1,
\end{equation}

\begin{equation}
\mathcal{R}^{(i)}(x_1^1) = \frac{(1-\delta_i)}{1+e^{\epsilon_i}}\mathcal{Q}_1^0+\frac{(1-\delta_i)e^{\epsilon_i}}{1+e^{\epsilon_i}}\mathcal{Q}_1^1 + \delta_i \mathcal{Q}_1.
\end{equation}
$\forall x_i \in \{x_2, \cdots, x_n\},$
\begin{equation}\label{DecomposXi}
 \mathcal{R}(x_i) = \frac{1-\delta_i}{1+e^{\epsilon_i}}\mathcal{Q}^0_1+\frac{1-\delta_i}{1+e^{\epsilon_i}}\mathcal{Q}^1_1 +\left(1-\frac{2(1-\delta_i)}{1+e^{\epsilon_i}}\right)\mathcal{Q}^i.
\end{equation}

\end{lemma}
\begin{proof}
For inputs $X_0 = \{x_1^0, x_2, \ldots, x_n\}$ and $X_1 = \{x^1_1, x_2, \ldots, x_n\}$, $\mathcal{R}^{(i)}$ satisfies the constraints of Lemma \ref{DPtransform},
so there exists an
$(\epsilon_i,\delta_i)$-PLDP local randomizer $\mathcal{R'}: \mathcal{D} \rightarrow \mathcal{Z}$ for the $i$-th output and post-processing function $proc(\cdot)$ such that $proc(\mathcal{R'}^{(i)}(x)) = \mathcal{R}^{(i)}(x)$, and
$$
P(\mathcal{R'}^{(i)}(x_1^0)=z) = \left\{
    \begin{array}{ll}
        0 & \mbox{if } z=A, \\
        \frac{(1-\delta_i)e^{\epsilon_i}}{1+e^{\epsilon_i}} & \mbox{if } z = 0, \\
        \frac{1-\delta_i}{1+e^{\epsilon_i}} & \mbox{if } z = 1, \\
        \delta_i & \mbox{if } z = B.
    \end{array}
\right.
$$

$$
P(\mathcal{R'}^{(i)}(x_1^1)=z) = \left\{
    \begin{array}{ll}
        \delta_i & \mbox{if } z=A, \\
        \frac{1-\delta_i}{1+e^{\epsilon_i}} & \mbox{if } z = 0, \\
        \frac{(1-\delta_i)e^{\epsilon_i}}{1+e^{\epsilon_i}} & \mbox{if } z = 1, \\
        0 & \mbox{if } z = B.
    \end{array}
\right.
$$
Let $L = \{z\in \mathcal{Z}|\mathbb{P}(\mathcal{R'}(x_1^0=z))=\frac{(1-\delta_i)e^{\epsilon_i}}{1+e^{\epsilon_i}}$ and $\mathbb{P}(\mathcal{R'}(x_1^1=z))=\frac{1-\delta_i}{1+e^{\epsilon_i}}\}$ ,
$U = \{z\in \mathcal{Z}|\mathbb{P}(\mathcal{R'}(x_1^1=z))=\frac{1-\delta_i}{1+e^{\epsilon_i}}$ and
$\mathbb{P}(\mathcal{R'}(x_1^1=z))=\frac{(1-\delta_i)e^{\epsilon_i}}{1+e^{\epsilon_i}}\}$. Let $M = \mathcal{Z}/\{L\bigcup U\}$ and
$p = \sum_{z \in \mathcal{L}} p_z = \sum_{z \in \mathcal{U}} p_z.$ Since conditioned on the output lying in $\mathcal{L}$, the distribution of
$\mathcal{R}'(x_1^0)$ and $\mathcal{R}'(x_1^1)$ are the same. Let $\mathcal{W}_1^0 = \mathcal{R}'(x_1^0) |L = \mathcal{R}'(x_1^1)|L$,
$\mathcal{W}^1_1 = \mathcal{R}'(x_1^0)|U = \mathcal{R}'(x_1^1)|U$ and $\mathcal{W}_1 = \mathcal{R}'(x_1^0)|M =\mathcal{R}'(x_1^1)|M.$ Then
$$
\mathcal{R}'(x_1^0) = \frac{(1-\delta_i)e^{\epsilon_i}}{1+e^{\epsilon_i}}\mathcal{W}_1^0+\frac{1-\delta_i}{1+e^{\epsilon_i}}\mathcal{W}_1^1 + \delta_i \mathcal{W}_1,
$$

$$
\mathcal{R}'(x_1^1) = \frac{1-\delta_i}{1+e^{\epsilon_i}}\mathcal{W}_1^0+\frac{(1-\delta_i)e^{\epsilon_i}}{1+e^{\epsilon_i}}\mathcal{W}_1^1 + \delta_i \mathcal{W}_1.
$$
Further, for all $x_i \in \{x_2,\cdots, x_n\}$,
$$\mathcal{R}'(x_i) \ge \frac{1-\delta_i}{1+e^{\epsilon_i}}\mathcal{W}_1^0 + \frac{1-\delta_i}{1+e^{\epsilon_i}}\mathcal{W}_1^1 + \left(1-\frac{2(1-\delta_i)}{1+e^{\epsilon_i}}\right)\mathcal{W}_i.$$
Letting $\mathcal{Q}_1^0 = proc(\mathcal{W}_1^0)$, $\mathcal{Q}_1^1 = proc(\mathcal{W}_1^1)$, $\mathcal{Q}_1 = proc(\mathcal{W}_1)$ and for
all $i \in \{2, \cdots, n\},$ $\mathcal{Q}_i = proc(\mathcal{W}_i)$. The proof is completed.
\end{proof}

\begin{theorem}\label{thm:post3}
For a domain $\mathcal{D}$, if $\mathcal{A}_{R,S}(\mathcal{D})$ is a shuffled PLDP adaptive process, then for arbitrary two neighboring datasets $D_0,D_1 \in \mathcal{D}^n$ distinct at the $n$-th data point, there exists a post-processing function $proc(\cdot)$: $(0,1,2) \rightarrow \mathcal{S}^{(1)} \times \mathcal{S}^{(2)} \times \cdots \times \mathcal{S}^{(n)},$
such that $$T(\mathcal{A}_{R,S}(D_0),\mathcal{A}_{R,S}(D_1)) = T(proc(\rho_0),proc(\rho_1)).$$
Here,
\begin{equation}\label{rho0}
\rho_0 = (\Delta_0,\Delta_1,\Delta_2) + \pmb V,
\end{equation}
\begin{equation}\label{rho1}
\rho_1 = (\Delta_1,\Delta_0,\Delta_2) + \pmb V,
\end{equation}
$\Delta_2 \sim Bern(\delta_n), \Delta_0 \sim Bin(1-\Delta_2,\frac{e^{\epsilon_{n}}}{1+e^{\epsilon_{n}}})$,
$\Delta_1 = 1 - \Delta_0 - \Delta_2$, where
$Bern(p)$ denotes a Bernoulli random variable with bias $p$, $Bin(n,p)$ denotes a Binomial distribution with $n$ trials and success probability $p$. In addition, $\pmb V = \sum_{i=1}^{n-1}MultiBern\left(\frac{1-\delta_{i}}{1+e^{\epsilon_i}}, \frac{1-\delta_i}{1+e^{\epsilon_i}}, 
1-\frac{2(1-\delta_i)}{1+e^{\epsilon_i}}
\right)$, where
 $MultiBern(\theta_1,\cdots, \theta_d)$ represents a $d$-dimensional Bernoulli distribution with $ \sum_{j=1}^d\theta_j = 1$.
\end{theorem}

In order to enhance readability, proof details are placed in the section of Appendix. Based on Theorem \ref{thm:post3}, we can simplify the original problem by analyzing the shuffling process in a simple non-adaptive protocol. 

The primary objective in the following is to demonstrate the distance between two distributions. The Berry Esseen lemma \cite{berry1941accuracy,Esseen1942} is highly valuable and essential for proving asymptotic properties.
\begin{lemma}[Berry Esseen]\label{BerryBound}
Let $P = (\xi_0,\xi_1,\xi_2) \sim  \sum_{i=1}^m MultiBern\left(\frac{p_i}{2},\frac{p_i}{2},1-p_i\right)$ and $Q \sim  N(\mu,\Sigma)$,
where $\mu = \mathbb{E}(P)$ and $\Sigma = Var(P).$ Then for the first two components $(X_0,X_1)$, there exists
$C>0$, such that $\|\tilde{P}-\tilde{Q}\|_{TV} \le \frac{C}{\sqrt{m}}$, where $\tilde{P}$ and $\tilde{Q}$ represent the distribution of $(\xi_0,\xi_1)$ and corresponding
normal distribution, respectively.
\end{lemma}

In fact, for given $n$, we can obtain sophisticated bound of $\|\tilde{P}-\tilde{Q}\|_{TV}$ by numerical methods. Without loss of generality, we assume $\epsilon_i=\epsilon_0$, $\delta_i = \delta_0=O(1/n)$, then
$p_0 = \frac{1-\delta_0}{1+e^{\epsilon_0}}$. For some fixed output $(\xi_0,\xi_1)=(k_0,k_1)$, we approximate by integrating the normal probability density function around that point.
Let $G(\cdot)$ be the cumulative distribution function of $\tilde{Q}$ and  $h(k_0,k_1) = G(k_0+0.5,k_1+0.5)-G(k_0+0.5,k_1)-G(k_0,k_1+0.5)+G(k_0,k_1)$, then
\begin{equation}
\|\tilde{P}-\tilde{Q}\|_{TV} = \sup_{(k_0,k_1)}|\mathbb{P}(\xi_0=k_0, \xi_1 = k_1)-h(k_0,k_1)|.
%\mathbb{P}(X=k)=\int_{k-0.5}^{k+0.5}\frac{1}{\sqrt{2\pi np(1-p)}}e^{-\frac{(x-np)^2}{2np(1-p)}}dx
\end{equation}
% For example, if we take $n=100, \epsilon_i = \epsilon_0=0.1, \delta_i = \delta_0 = 1/n$, numerical method gives $\|\tilde{P}-\tilde{Q}\|_{TV}=0.0206$, which is nearly $O(1/n)$.

\begin{lemma}\label{NormalT}
Let $p_i = \frac{2(1-\delta_i)}{1+e^{\epsilon_i}}$, if $\bar{\mu} = \sum_{i=1}^{n-1}(\frac{p_i}{2}, \frac{p_i}{2})'$ and $\mu_0 = (1,0)'+\bar{\mu}$,
$\mu_1 = (0,1)'+\bar{\mu}$, then
$T(N(\mu_0,\pmb{\Sigma}),N(\mu_1,\pmb{\Sigma})) = \Phi(\Phi^{-1}(1-\alpha)-\frac{2}{\sqrt{\sum_{i=1}^{n-1}}p_i})$,
$$\pmb{\Sigma} = \sum_{i=1}^{n-1} \left(
\begin{array}{cc}
\frac{p_i}{2}(1-\frac{p_i}{2}) & -\frac{p_i^2}{4} \\
-\frac{p_i^2}{4} & \frac{p_i}{2}(1-\frac{p_i}{2}) \\
\end{array}
\right).$$
\end{lemma}

\begin{theorem}[Enhanced Central Privacy Upper bound] \label{thm:GDPdistance}
Assume $\rho_0$ and $\rho_1$ are defined in equations (\ref{rho0}) and (\ref{rho1}), then there exists $C>0$, such that
\begin{equation}
T(\rho_0,\rho_1) \ge \left(G_\mu\left({\alpha+\frac{C}{\sqrt{n-1}}}\right)-\frac{C}{\sqrt{n-1}}\right)
,
\end{equation}
where
$G_\mu(\alpha) = \Phi(\Phi^{-1}(1-\alpha)-\mu), 
\mu = \sqrt{\frac{2}{\sum_{i=1}^{n} \frac{1-\delta_i}{1+e^{\epsilon_i}}-\max_{i}{\frac{1-\delta_{i}}{1+e^{\epsilon_{i}}}}}}.
$
In an unambiguous manner, we refer to it as approximately following the $\mu$-GDP.
\end{theorem}
\begin{proof}
First, let's analyze the scenarios where the $n$-th data point differs.
According to the definition of $(\Delta_0,\Delta_1,\Delta_2)$,
\begin{equation}
(\Delta_0,\Delta_1,\Delta_2) =
\begin{cases}
    \begin{array}{lll}
        (0,0,1) & \quad w.p. & \quad \delta_{n}; \\
        (1,0,0) &  \quad w.p. & \quad (1-\delta_{n})\frac{e^{\epsilon_{n}}}{1+e^{\epsilon_{n}}}; \\
        (0,1,0) & \quad  w.p. & \quad (1-\delta_{n})\frac{1}{{1+e^{\epsilon_{n}}}}.
    \end{array}
\end{cases}
\end{equation}
When $\Delta_2 = 1$, $\rho_0$ and $\rho_1$ are  indistinguishable, which indicates that $T(\rho_0,\rho_1)|_{\Delta_2=1} = 1-\alpha$.
 Let $\rho'_0=(1,0,0)'+\sum_{i=1}^{n-1}MultiBern\left(p_i/2,p_i/2,1-p_i\right)$ and $\rho'_1=(0,1,0)'+\sum_{i=1}^{n-1}MultiBern\left(p_i/2,p_i/2,1-p_i\right)$ with $p_i =
 \frac{1-\delta_i}{1+e^{\epsilon_i}}
 $, then
\begin{equation}
T(\rho_0,\rho_1) = \delta_{n}(1-\alpha)+(1-\delta_{n})T_{symm}(\rho'_0,\rho'_1),
\end{equation}
where $T_{symm}(\rho'_0,\rho'_1) = \max \{T(\rho'_0,\rho'_1),T(\rho'_1,\rho'_0)\}$.
Assume $P \sim N(\mu_0,\Sigma)$, $Q \sim N(\mu_1,\Sigma)$, where $\mu_0,\mu_1,\Sigma$ are same as that in Lemma \ref{NormalT}.
%then
%\begin{equation}
%T(\rho_0,\rho_1) = \\
%\frac{(1-\delta_{\max})e^{\epsilon_{\max}}}{1+e^{\epsilon_{\max}}}T(P,Q)
%+ \frac{1-\delta_{\max}}{1+e^{\epsilon_{\max}}}T(Q,P)
%+ \delta_{\max}(1-\alpha).
%\end{equation}
%Due to the symmetry of normal distribution, we can obtain that $T(P,Q)=T(Q,P)$.
Let $\mu = \sqrt{\frac{2}{\sum_{i=1}^{n-1} \frac{1-\delta_i}{1+e^{\epsilon_i}}}}$,
according to equation (\ref{Tv2beta}), $$T(\rho'_0,P) \ge 1- \alpha - \|\rho'_0-P\|_{TV},$$
$$T(\rho'_1,Q) \ge 1- \alpha - \|\rho'_1-Q\|_{TV},$$
then based on Fact 4, $$T(\rho'_0,Q) \ge \Phi(\Phi^{-1}(1- \alpha-\|\rho'_0-P\|_{TV})-\mu)= F(\alpha).$$
Reusing Fact 4, we can obtain that
\begin{eqnarray}    \label{eq}
T(\rho'_0,\rho'_1)&\ge&1-(1-F(\alpha))- \|\rho'_1-Q\|_{TV} \nonumber    \\
~&=&F(\alpha)- \|\rho'_1-Q\|_{TV} \nonumber    \\
\end{eqnarray}
%
%$$
%T(\rho'_0,\rho'_1) \ge 1-(1-F(\alpha))- \|\rho'_1-Q\|_{TV} =
%\Phi(\Phi^{-1}(1-\alpha-\|\rho'_0-P\|_{TV})-\mu)- \|\rho'_1-Q\|_{TV}.
%$$
Lemma \ref{BerryBound} shows that there exists $C>0$, such that $\|\rho'_1-Q\|_{TV} \le \frac{C}{\sqrt{n-1}}$
and $\|\rho'_0-P\|_{TV} \le \frac{C}{\sqrt{n-1}}$. Hence
$$
T(\rho'_0,\rho'_1) \ge G_\mu\left({\alpha+\frac{C}{\sqrt{n-1}}}\right)-\frac{C}{\sqrt{n-1}}.
$$
Then
\begin{equation*}
\begin{split}
T(\rho_0,\rho_1) &\ge \delta_{n}(1-\alpha) \\
&\quad +(1-\delta_{n})\left(G_\mu\left({\alpha+\frac{C}{\sqrt{n-1}}}\right)-\frac{C}{\sqrt{n-1}}\right).
\end{split}
\end{equation*}
Since for an arbitrary trade-off function $f$, we have $f \le 1-\alpha$, it follows that:
$$
T(\rho_0,\rho_1) \ge \left(G_\mu\left(\alpha+\frac{C}{\sqrt{n-1}}\right)-\frac{C}{\sqrt{n-1}}\right).
$$
Finally, taking into account the case where the $i$-th ($1\le i \le n$) data  differs in neighboring datasets, the privacy bound is determined based on the worst-case scenario, that is,
$
\mu = \sqrt{\frac{2}{\sum_{i=1}^{n} \frac{1-\delta_i}{1+e^{\epsilon_i}}-\max_{i}{\frac{1-\delta_{i}}{1+e^{\epsilon_{i}}}}}}.
$
% Furthermore, according to the monotonicity of $\mu$-GDP, this inequality also holds for slightly bigger $\mu$ with
% $\mu = \sqrt{\frac{2}{\sum_{i=1}^{n} \frac{1-\delta_i}{1+e^{\epsilon_i}}}}.
% $
\end{proof}
\subsubsection*{Comparison with Existing Results}
We provide numerical evaluations for privacy amplification effect under fixed LDP settings in Table \ref{tab:LDPbudgets}.
Given a local privacy budget set $\epsilon^\ell \in [0.01,2]$.
For the purpose of comparison, we examine privacy amplification for a fixed $\epsilon^\ell$ while varying $n$ from $10^3$ to $10^4$, with central $\delta$ for shuffling to be $10^{-4}$ for the sake of simplicity. To avoid misunderstandings, we repeat the first $10^3$ parameters. Considering that convergence rate in Lemma \ref{BerryBound} is nearly $O(1/n)$ and can be negligible in numerical analysis, our focus lies in measuring $G_\mu$.

\begin{table}
  \centering
  % and 
  % $50\%~0.5+50\%~0.01$ means that half of the privacy budgets are $0.5$ and the other half are $0.01$.
      
      \setlength{\tabcolsep}{3pt} % 调整列距
    \begin{tabular}{lc}
    \toprule
    Name  & Distribution of $\epsilon^\ell = (\epsilon_1^\ell,\cdots, \epsilon_n^\ell)$  \\
  \midrule
    Unif 1 & $U(0.01,1)$  \\

    Unif 2 & $U(0.01,2)$ \\
    Constant& $0.5$  \\
     Mixed Constant & $50\%~ 0.5+50\% ~0.01$  \\

    \bottomrule
    \end{tabular}%
    \caption{Distributions of LDP budgets $\epsilon^\ell$. $U(a,b)$ represents uniform distribution ranging from $a$ to $b$.}
    \label{tab:LDPbudgets}
\end{table}%

To keep it concise, we use Fact 3 in Appendix to compute the corresponding central $\epsilon$ and $\delta$ for Theorem \ref{thm:GDPdistance}.
Baseline bounds of privacy amplification effect include: [Liu23] \cite{Liu2023}, [FMT22] \cite{Feldman2022hiding}, [Erlingsson19] \cite{Erlingsson2019amplification}.
[Liu23] provides bounds for the personalized scenario, while [FMT22] and [Erlingsson19] only consider the same $\epsilon^\ell$.

 The numerical results demonstrate the following results:
(i) Our bound is suitable for extreme privacy budgets while [Liu23] required each $\epsilon_i$ should not be close to zero. However, it is natural to encounter user responses that contain no information, resulting in $\epsilon_i=0$.
(ii) As the sample size $n$ increases, the amplification effect also increases proportionally to the square root of $n$.
(iii) Our privacy bounds significantly outperform in all current scenarios, even in cases where the privacy parameters are the same.
\begin{figure}[htbp]
  \centering
  \includegraphics[scale=0.35]{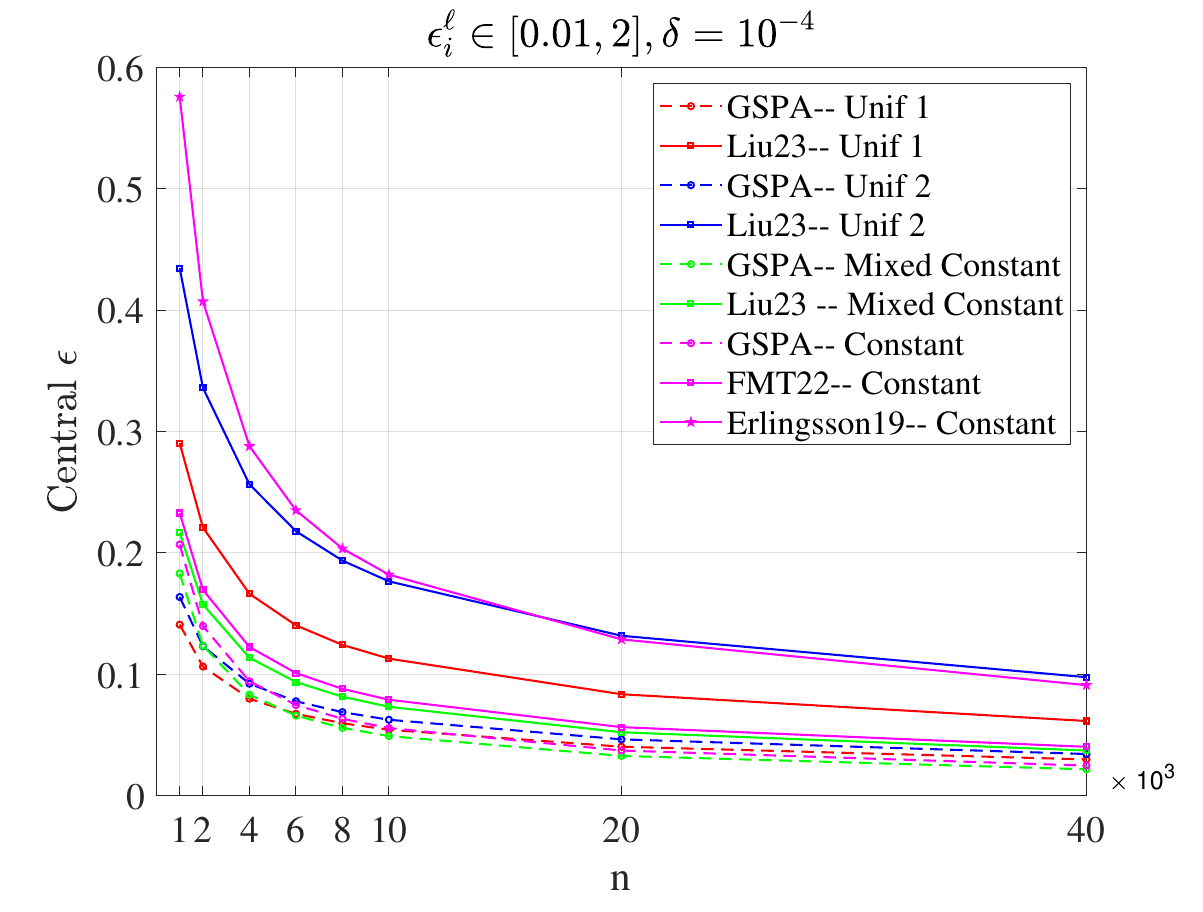}\\
  \caption{Privacy Bounds for Varied Budgets}\label{Fig:VsLiu}
\end{figure}

\section*{Application and Experiments}
All the experiments are implemented on a workstation with an Intel Core i5-1155G7 processor on Windows 11 OS.
\subsection{Application to Mean and Frequency Estimation}

\subsubsection{Mean Estimation}
The average function is a fundamental and commonly used mathematical operation with wide-ranging applications. In this section, we apply GSPA to  the average function on the synthetic data.
We randomly divide the users into three groups: conservative, moderate, and liberal. The fraction of three groups are determined by $f_c, f_m, f_l$.
As is reported \cite{Acquisti2005privacy}, the default values in this experiment are $f_c=0.54, f_m =0.37, f_l=0.09$. For convenience, the privacy preferences for the users in conservative,
moderate and liberal groups are $\epsilon_C, \epsilon_M$ and $\epsilon_l$, respectively. In the LDP case, the privacy preference of users in the liberal group is fixed at $\epsilon_L=1$, while the default values of $\epsilon_C$ and $\epsilon_M$ are set to $0.1$ and $0.5$, respectively. 
\begin{theorem}\label{coro:LapsumGDP}
Algorithm \ref{alg:LapSum} approximately preserves $\mu$-GDP for each user, where $
\mu = \sqrt{\frac{2}{\sum_{i=1}^{n} \frac{1-\delta_i}{1+e^{\epsilon_i}}-\max_{i}{\frac{1-\delta_{i}}{1+e^{\epsilon_{i}}}}}}.
$
\end{theorem}
\begin{proof}
According to the definition of Laplace mechanism, data point $i \in [n]$ satisfies $(\epsilon,0)$-LDP. Combined with
Theorem \ref{thm:GDPdistance}, we can obtain that Algorithm \ref{alg:LapSum} approximately satisfies $\mu$-GDP with $
\mu = \sqrt{\frac{2}{\sum_{i=1}^{n} \frac{1-\delta_i}{1+e^{\epsilon_i}}-\max_{i}{\frac{1-\delta_{i}}{1+e^{\epsilon_{i}}}}}}.
$
\end{proof}

Next, we simulate the accuracy for different set of privacy protection. To facilitate comparison, we set $f_l=0.09$ as a fixed value and vary $f_c$ from $0.01$ to $0.5$ with $f_m=1-f_l-f_c$. Additionally, we generate $n=10,000$ privacy budgets for users based on the privacy preferences rule.
We assume that each sample is drawn from a normal distribution $N(50,\sigma^2)$, and then the samples are clipped into the range $[20, 80]$. We repeat this procedure for a total of $1,000$ times to give a confidence interval.
%需要修改润色
According to Fact \ref{muexpand}, privacy parameter $\mu$ under the shuffle model can be obtained for varying $\epsilon_c$.
Figure \ref{fig:maevsfc} shows that an increase in the proportion of conservative users leads to a decrease in estimation accuracy. On the other hand, Figure \ref{fig:maevsepsilon_c} demonstrates that increasing privacy budget is beneficial for improving accuracy.
\begin{algorithm}[htbp]
    % \caption{Shuffled PLDP for average estimation}
     \caption{Mean estimation with GSPA.}
    \label{alg:LapSum}
    \renewcommand{\algorithmicrequire}{\textbf{Input:}}
    \renewcommand{\algorithmicensure}{\textbf{Output:}}
    \begin{algorithmic}[1]
        \REQUIRE Dataset $X = (x_1,\ldots,x_n)\in \mathbb{R}^n$, privacy budget $\mathcal{S} = \{\epsilon_1, \cdots, \epsilon_n\}$ for each user.
        \ENSURE $z \in \mathbb{N}$   %%output
        \FOR{each $i$ $\in [n]$}
        \STATE $y_i \leftarrow x_i+Lap(\Delta f/\epsilon_i)$
        \ENDFOR
        \STATE Choose a random permutation $\pi$: $[n] \rightarrow [n]$
        \STATE $z =\frac{1}{n} \sum_{i=1}^n y_{\pi(i)}$
        \RETURN $z$
    \end{algorithmic}
\end{algorithm}

%\begin{figure}
%  \centering
%  % Requires \usepackage{graphicx}
%  \includegraphics[scale=0.4]{picture/MAEvsfc.pdf}\\
%  \caption{Impact of $f_c$}\label{Fig:MAEvsfc}
%\end{figure}
%
%\begin{figure}
%  \centering
%  % Requires \usepackage{graphicx}
%  \includegraphics[scale=0.4]{picture/MAEvsEpsilon_c.pdf}\\
%  \caption{Impact of $\epsilon_c$ (with $\epsilon_m = 0.5$)}\label{Fig:MAEvsEpsilon_c}
%\end{figure}
%\begin{figure}[htbp]
%  \centering
%  \begin{minipage}[b]{0.23\textwidth}
%    \centering
%    \includegraphics[width=\textwidth]{picture/MAEvsfc.pdf}
%    \caption{Impact of $f_c$}
%    \label{fig:maevsfc}
%  \end{minipage}
%  \hfill
%  \begin{minipage}[b]{0.23\textwidth}
%    \centering
%    \includegraphics[width=\textwidth]{picture/MAEvsEpsilon_c.pdf}
%    \caption{Impact of $\epsilon_c$ }
%    \label{fig:maevsepsilon_c}
%  \end{minipage}
%\end{figure}

\begin{figure}[htbp]
  \centering
  \begin{subfigure}[b]{0.23\textwidth}
    \centering
    \includegraphics[width=\textwidth]{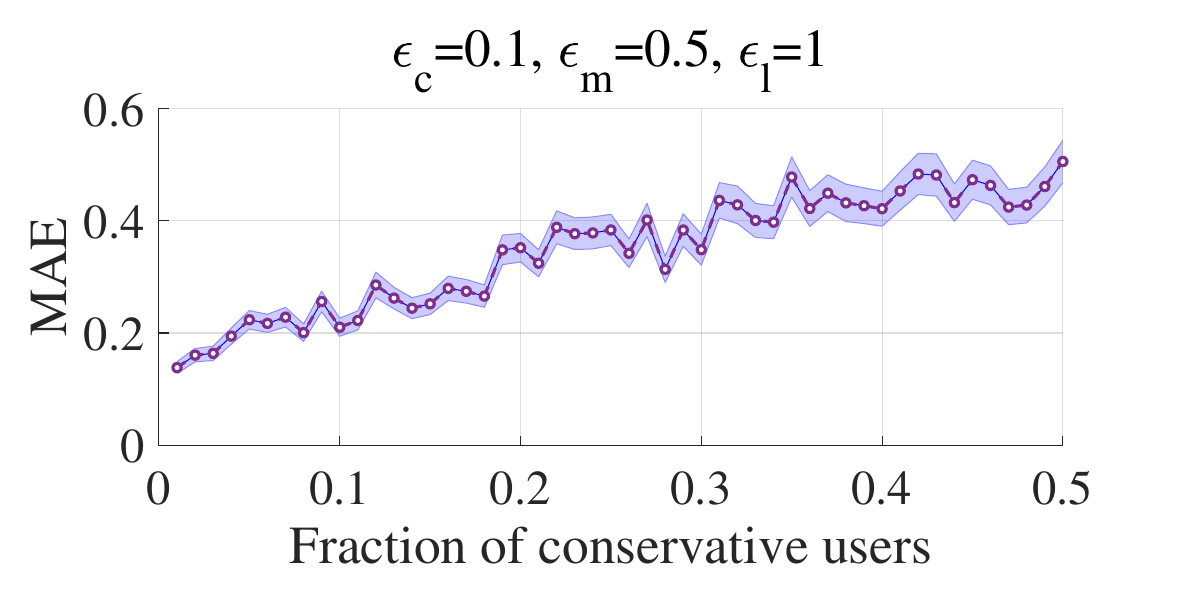}
    \caption{{Impact of $f_c$ (Mean)}}
    \label{fig:maevsfc}
  \end{subfigure}
  \hfill
  \begin{subfigure}[b]{0.23\textwidth}
    \centering
    \includegraphics[width=\textwidth]{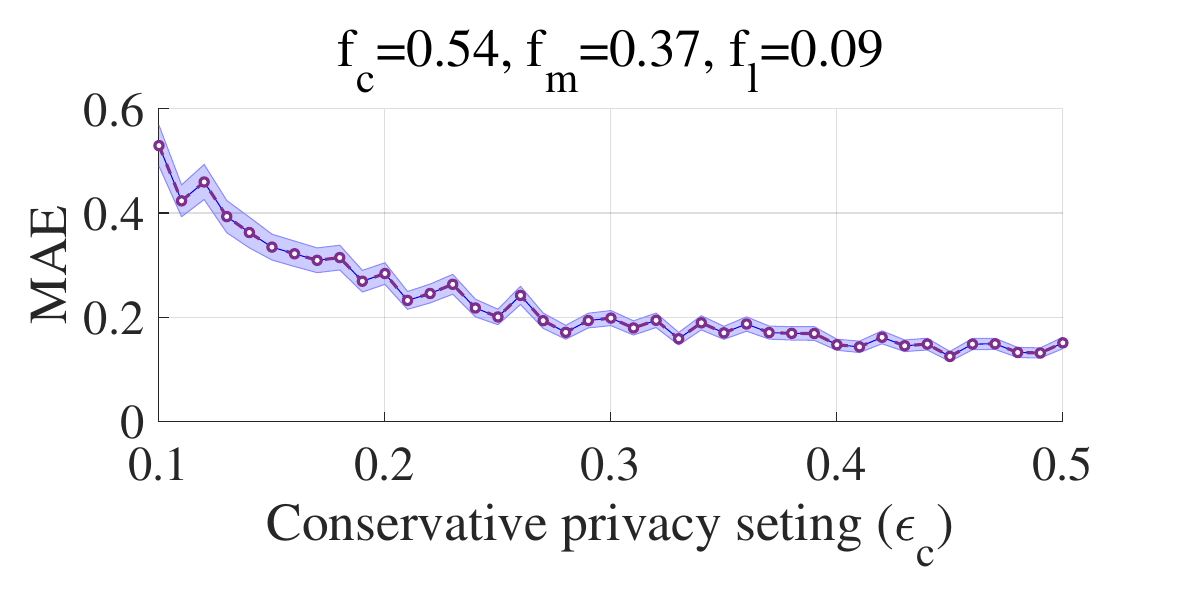}
    \caption{Impact of $\epsilon_c$ (Mean)}
    \label{fig:maevsepsilon_c}
  \end{subfigure}

  \vspace{0.5cm} % 调整垂直间距

  \begin{subfigure}[b]{0.23\textwidth}
    \centering
    \includegraphics[width=\textwidth]{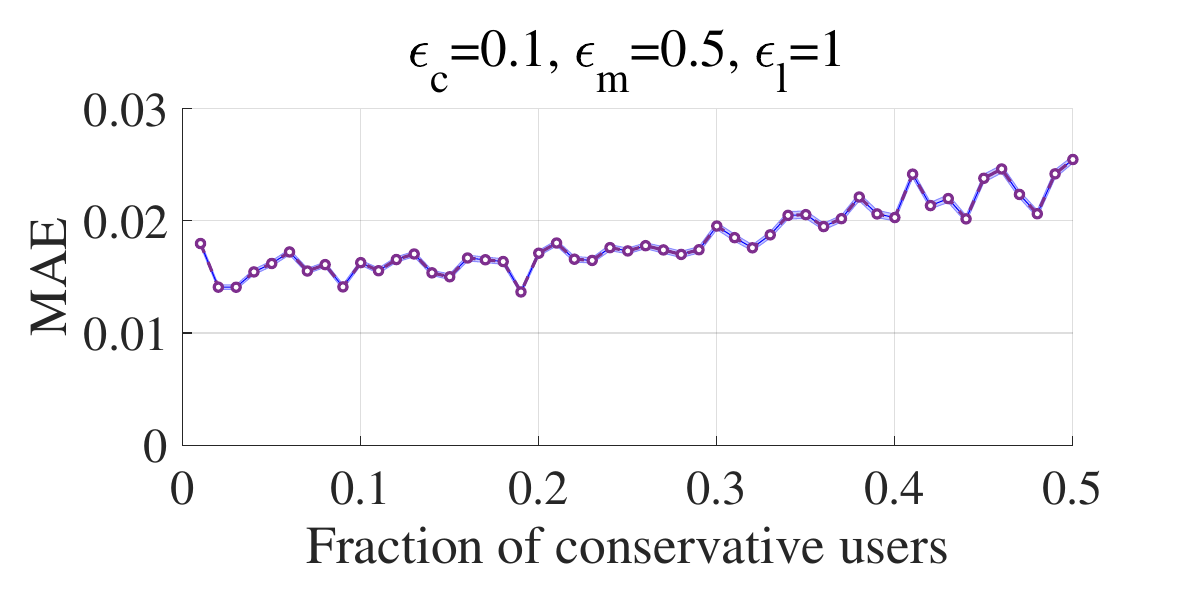}
    \caption{Impact of $f_c$ (Frequency)}
    \label{fig:sub3}
  \end{subfigure}
  \hfill
  \begin{subfigure}[b]{0.23\textwidth}
    \centering
    \includegraphics[width=\textwidth]{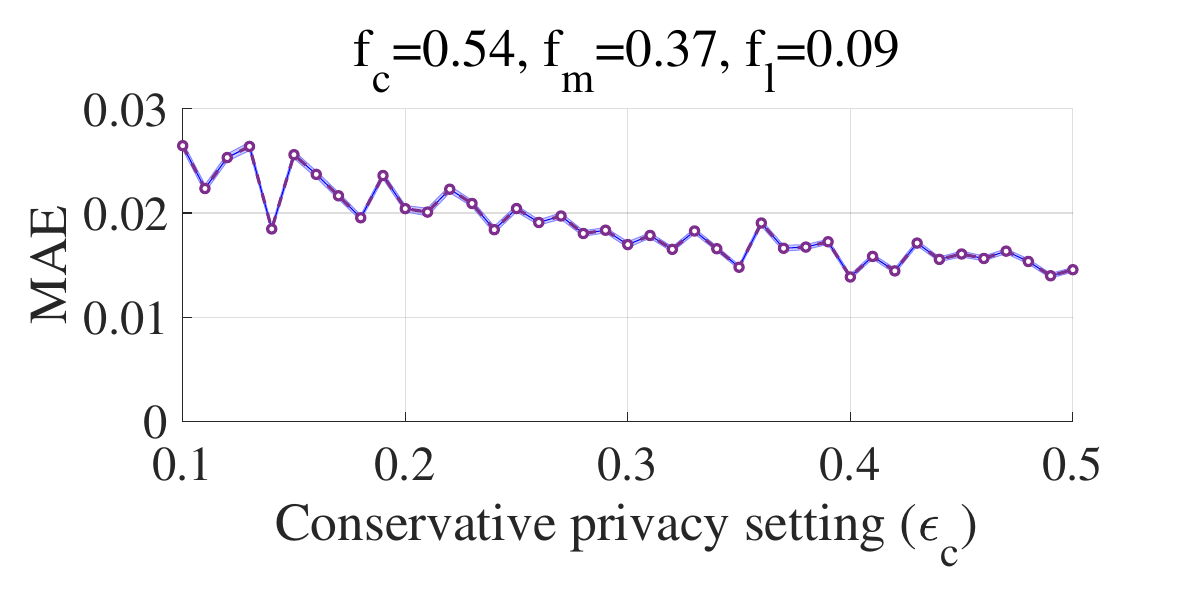}
    \caption{Impact of $\epsilon_c$ (Frequency)}
    \label{fig:sub4}
  \end{subfigure}
  \caption{Impact of privacy parameter settings on MAE.}
  \label{fig:full}
\end{figure}
\subsubsection{Frequency estimation}
In machine learning, frequency estimation is often used as a preprocessing step to understand the distribution and importance of different features or categories within a dataset. By accurately estimating the frequencies of various features or categories, it helps in feature selection, dimensionality reduction, and building effective models. 

In order to obtain the dataset, a total of 10,000 records are generated for counting. Each record is encoded as a binary attribute. The proportion of records with a value of $1$ is determined by a density parameter $c$, which ranges from $0$ to $1$ (with a default value of $c$ = 0.7).
\begin{theorem}\label{coro:ExpSumGDP}
Algorithm \ref{alg:ExpSum} approximately preserves $\mu$-GDP for each user, where $
\mu = \sqrt{\frac{2}{\sum_{i=1}^{n} \frac{1-\delta_i}{1+e^{\epsilon_i}}-\max_{i}{\frac{1-\delta_{i}}{1+e^{\epsilon_{i}}}}}}.
$
\end{theorem}
The proof of Theorem \ref{coro:ExpSumGDP} is the same as Theorem \ref{coro:LapsumGDP}. The direct calculation shows that $z$ is an unbiased estimator of $c$, that is, $\mathbb{E}(z)=c$. Similar to the average function, we adopt the same configuration for personalized privacy budgets.

\begin{algorithm}[htbp]
    \caption{Frequency estimation with GSPA}
    \label{alg:ExpSum}
    \renewcommand{\algorithmicrequire}{\textbf{Input:}}
    \renewcommand{\algorithmicensure}{\textbf{Output:}}
    \begin{algorithmic}[1]
        \REQUIRE Dataset $X = (x_1,\ldots,x_n)\in \{0,1\}^n$, privacy budget $\mathcal{S} = \{\epsilon_1, \cdots, \epsilon_n\}$ for each user.
        \ENSURE $z \in \mathbb{N}$   %%output
        \FOR{each $i$ $\in [n]$}
        \IF{$x_i = 1$}
        \STATE $y_i \leftarrow Ber(\frac{e^{\epsilon_i}}{1+e^{\epsilon_i}})$
        \ELSE
        \STATE $y_i \leftarrow Ber(\frac{1}{1+e^{\epsilon_i}})$
        \ENDIF
        \ENDFOR
        \STATE Choose a random permutation $\pi$: $[n] \rightarrow [n]$
        \STATE $A =  \sum_{i=1}^n y_{\pi(i)}$
        \STATE $B = \sum_{i=1}^n \frac{1}{1+e^{\epsilon_{\pi(i)}}}$
        \STATE $z =\frac{A-B}{n-2B} $
        \RETURN $z$
    \end{algorithmic}
\end{algorithm}

%\begin{figure}
%  \centering
%  % Requires \usepackage{graphicx}
%  \includegraphics[scale=0.3]{picture/MAEvsEpsilon_c_freq.pdf}\\
%  \caption{Impact of $\epsilon_c$ for frequency estimation}\label{fig:MAEvsEPFreq}
%\end{figure}

\subsection{Personalized Private Stochastic Gradient Descent}
The private stochastic gradient descent is a common method in deep learning \cite{abadi2016deep}. 
However, personalized private stochastic gradient descent  combines personalized differential privacy  with stochastic gradient descent optimization for model training and parameter updates while ensuring privacy protection.
In the context of personalized differential privacy, privacy of individual users must be protected, and direct use of raw data for parameter updates is not feasible .

The key idea of personalized differential privacy is to introduce personalized parameters into the differentially private mechanism to flexibly adjust the level of privacy protection. For the gradient descent algorithm, personalized differential privacy can be achieved by introducing noise during gradient computation.

\begin{theorem}\label{ProSGD1}
Algorithm \ref{alg:SGD} approximately satisfies $\sqrt{T}\mu$-GDP with $
\mu = \sqrt{\frac{2}{\sum_{i=1}^{n} \frac{1-\delta_i}{1+e^{\epsilon_i}}-\max_{i}{\frac{1-\delta_{i}}{1+e^{\epsilon_{i}}}}}}.
$
\begin{proof}
For arbitrary $j \in [m]$, client $j$ satisfies $(\epsilon_j,\delta_j)$-LDP before sending to the shuffler by the definition of Gaussian mechanism. By using Theorem \ref{thm:GDPdistance}, it preserves $\mu$-GDP after shuffling with $
\mu = \sqrt{\frac{2}{\sum_{i=1}^{n} \frac{1-\delta_i}{1+e^{\epsilon_i}}-\max_{i}{\frac{1-\delta_{i}}{1+e^{\epsilon_{i}}}}}}.
$ In addition, combined with Fact \ref{muexpand}, it holds $\sqrt{T}\mu$-GDP under $T$-fold composition.

\end{proof}
%Algorithm \ref{alg:SGD} satisfies $F \otimes f_{0,\delta'}$-DP, where $F$ is %$F=\Phi(\Phi^{-1}(1-\alpha)-\frac{2}{\sqrt{m-1}} e^{\frac{\epsilon_0}{2}})$ and %$\delta'=1-(1-(1+\frac{1}{2e^{\epsilon_0}})\delta_0)^m$.
\end{theorem}
\subsubsection*{Dataset and implementation}
The MNIST dataset \cite{lecun1998gradient} for handwritten digit recognition consists of $60,000$ training images and $10,000$ test images. Each sample in the dataset represents a $28 \times 28$ vector generated from handwritten images, where the independent variable corresponds to the input vector, and the dependent variable represents the digit label ranging from $0$ to $9$.
In our experiments, we consider a scenario with $m$ clients, where each client has $n/m$ samples. For simplicity, we train a simple classifier using a feed-forward neural network with ReLU activation units and a softmax output layer with $10$ classes, corresponding to the $10$ possible digits. The model is trained using cross-entropy loss and an initial PCA input layer with $60$ components. At each step of the shuffled SGD, we choose at one client at random without replacement. The parameters of experimental setup is listed in Table \ref{tab:Experimentsetup}.
This experiment is designed to demonstrate the use cases of the shuffle model and therefore does not focus on comparing with previous results. For comparative results, please refer to Figure 2.

\subsubsection*{Parameter Selection}
As a result, our approach achieves an accuracy of $96.78 \%$ on the test dataset after approximately $50$ epochs. This result is consistent with the findings of a vanilla neural network \cite{lecun1998gradient} trained on the same MNIST dataset.
By employing this methodology, we can effectively train a simple classifier that achieves high accuracy in recognizing handwritten digits from the MNIST dataset.

\begin{table}[htbp]
\setlength{\tabcolsep}{3pt} % 调整列距
  \centering
    \begin{tabular}{lll}
    \toprule
    \textbf{Parameters} & \multicolumn{1}{l}{\textbf{Value}} & \textbf{Explanation} \\
    \midrule
    $C$     & $2$    & Clipping bound \\
    $\delta^{\ell}$ & \multicolumn{1}{l}{$10^{-5}$} & Indistinguishability parameter \\
    $\epsilon^{\ell}$ & \multicolumn{1}{l}{$[0.01,2]$} & Privacy budget \\
    $\eta$ & $0.05$   & Step size of the gradient \\
    $m$     & $100$    & The number of clients \\
    $n$    & $60,000$ & Total number of training samples \\
    $T$    & $50$     & The number of epochs \\
    \bottomrule
    \end{tabular}%
    \caption{Experiment Setting for the Shuffled Personalized-SGD on the MNIST dataset.}
  \label{tab:Experimentsetup}%
\end{table}%
\begin{figure}
  \centering
  \includegraphics[scale=0.4]{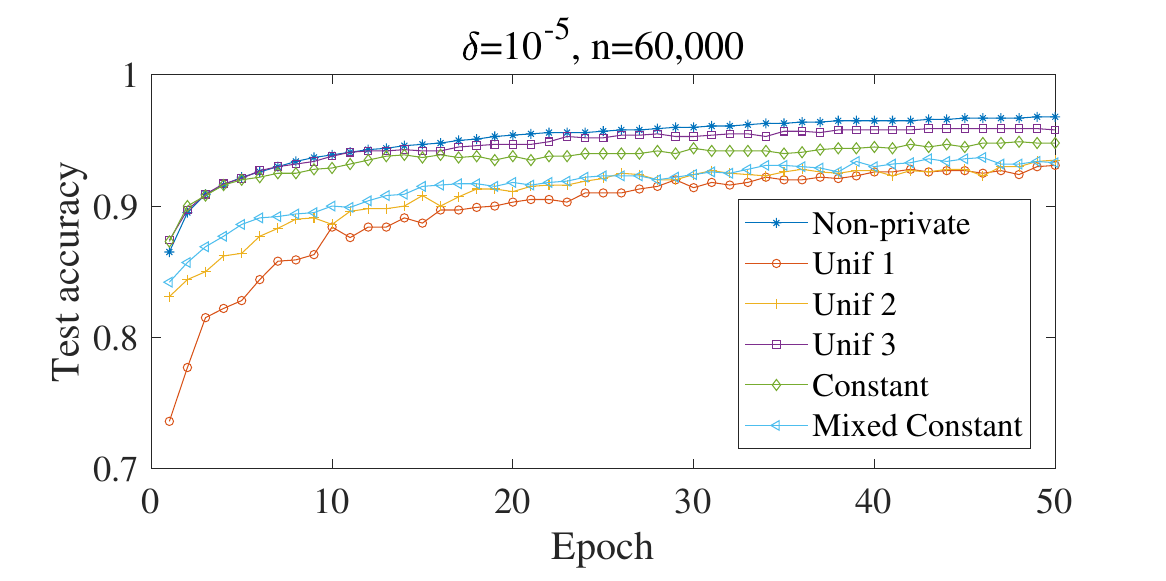}\\
  \caption{Comparison of Test Accuracy with Different Noise Distributions}\label{Fig:MnistSGDTrain}
\end{figure}
\subsubsection*{Utility of GSPA}
We evaluate the utility of GSPA with $\epsilon^\ell$ drawing from Table \ref{tab:LDPbudgets} on MNIST dataset. We introduce Unif 3 as a distribution to represent $U(0.5,1)$.
The numerical results indicate that Unif 3 exhibits the best accuracy, which aligns with expectations as it corresponds to a larger value of the privacy budget. Despite constant scenario exhibiting stronger privacy protection than Unif 2, it actually achieves better accuracy. One possible reason behind this interesting observation is the significant difference in the privacy parameters, which can cause instability in gradient iterations.

\begin{algorithm}[h]
    \caption{SGD with GSPA}
    \label{alg:SGD}
    \renewcommand{\algorithmicrequire}{\textbf{Input:}}
    \renewcommand{\algorithmicensure}{\textbf{Output:}}
    \begin{algorithmic}[1]
        \REQUIRE Dataset $X = (x_1,\ldots,x_n)$, loss function $\mathcal{L}(\pmb \theta,x)$, initial point $\pmb \theta_0$, learning rate $\eta$, number of epochs $T$, privacy budget $\mathcal{S} = \{\epsilon_1, \delta_1, \cdots, \epsilon_n,\delta_n\}$,  batch size $m$ and gradient norm bound $C$. %and noise scale %$\sigma=\frac{2C}{m}\frac{\sqrt{2\log(1.25/\delta_0)}}{\epsilon_0}$.
        \ENSURE $\hat{\pmb \theta}_{T,m}$   %%output
        %\STATE $s \leftarrow s_0$
        \STATE
        Split $[n]$ into $n/m$ disjoint subsets $S_1, \cdots, S_m$ with equal size $m$
        \STATE Choose arbitrary initial point $ \hat{\pmb \theta}_{0}$
        \STATE Choose a random permutation $\pi$: $[m] \rightarrow [m]$
        \FOR{each $t \in [T]$}
        \STATE $\tilde{\pmb\theta}_{t,m} = \hat{\pmb \theta}_{t}$
        \FOR{each $i \in [n/m]$}
        \STATE $\sigma=\frac{2C}{m}\frac{\sqrt{2\log(1.25/\delta_{\pi(i)})}}{\epsilon_{\pi(i)}}$
            \STATE $\pmb b_i \sim N(0,\sigma^2 \pmb{I}_d)$
            \FOR{each $j \in S_{\pi(i)}$}
            \STATE \textbf{Compute gradient}:\\
            $\pmb g_i^{j} = \nabla \ell(\tilde{\pmb \theta}_{i-1},x_j)$
            \ENDFOR
            \STATE \textbf{Clip to norm $C$}:\\
            $\pmb g_i = \frac{\sum_j \pmb g_i^j}{m}$ \\
            $\tilde{\pmb g}_i =  {\pmb g_i} / \max(1,\| \pmb g_i \|_2/C)$
            
            \STATE $\tilde{\pmb \theta}_{i} = \tilde{\pmb \theta}_{i-1} - \eta ( \tilde{\pmb g}_i+\pmb b_i)$
        \ENDFOR
        \STATE $\hat{\pmb \theta}_{t,m} = \tilde{\pmb \theta}_m$
        \ENDFOR
        \RETURN $\hat{\pmb \theta}_{T,m}$
    \end{algorithmic}
\end{algorithm}
\section*{Conclusion and Future Work}
This work focuses on privacy amplification of shuffle model. To address the trade-off between privacy and utility, we propose the GSPA framework, which achieves a higher accuracy while maintaining at least 33\% privacy loss compared to existing methods.

In our future research, we plan to expand by incorporating additional privacy definitions such as Renyi differential privacy \cite{Girgis2021renyi}. 
Moreover, we acknowledge the significance of enhancing techniques for specific data types like images, speech, and other forms. This entails developing specialized privacy metrics, differential privacy mechanisms, and model training algorithms that offer more accurate and efficient privacy protection.
\section*{Appendix}
\renewcommand{\thetheorem}{\thesection.\arabic{theorem}}
\setcounter{theorem}{0}
% In probability theory, the following notations are commonly used:
% \begin{itemize}
%   \item $Bern(p)$ denotes a Bernoulli random variable with bias $p$. It is a single binary experiment that results in success with probability
% $p$ and failure with probability
% $1-p$.
%   \item $Bin(n,p)$ denotes a Binomial distribution with $n$ trials and success probability $p$.
%   \item $MultiBern(\theta_1,\cdots, \theta_d)$ represents a $d$-dimensional Bernoulli distribution with $ \sum_{j=1}^d\theta_j = 1$.
% \end{itemize}
\subsubsection{$f$-DP}
Here are several important properties of $f$-DP. We present these facts directly for the sake of brevity, and for comprehensive proofs, please refer to the related article \cite{Dong2022gaussian}.
\begin{fact}\label{factfunc}
$(\epsilon,\delta)$-DP is equivalent to $f_{\epsilon,\delta}$-DP, where
\begin{equation}
f_{\epsilon,\delta} = \max\{0,1-\delta-e^{\epsilon}\alpha,e^{-\epsilon}(1-\delta-\epsilon) \}.
\end{equation}
\end{fact}
\begin{fact}
$f$-DP holds the post-processing property, that is, if a mechanism $M$ is $f$-DP, then its post-processing $Proc \circ M$ is also $f$-DP.
\end{fact}
\begin{fact}\label{muexpand}($\mu$-GDP)
%For a symmetric trade-off function $f$, a mechanism is $f$-DP if and only if it is $(\epsilon,\delta(\epsilon))$-DP for all $\epsilon \ge 0$ with $\delta(\epsilon) = 1+f^*(-e^{-\epsilon})$, where $f^*(y)=\mathop{sup}\limits_{-\infty<x<\infty}yx-f(x)$.
A $f$-DP mechanism is called $\mu$-GDP if $f$ can be obtained by $f=T(N(0,1),N(\mu,1))=\Phi(\Phi^{-1}(1-\alpha)-\mu)$, where $\Phi(\cdot)$ is cumulative distribution function of standard Gaussian distribution $N(0,1)$.
 Then a mechanism is $\mu$-GDP if and only if it is $(\epsilon,\delta(\epsilon))$-DP for all $\epsilon \ge 0$, where $$
\delta(\epsilon)=\Phi(-\frac{\epsilon}{\mu}+\frac{\mu}{2})
-e^{\epsilon}\Phi(-\frac{\epsilon}{\mu}-\frac{\mu}{2}).
$$
 In
particular, if a mechanism is $\mu$-GDP, then it is $k\mu$-GDP for groups of size $k$ and the $n$-fold composition of $\mu_i$-GDP mechanisms is
$\sqrt{\mu_1^2+ \cdots + \mu_n^2}$-GDP.
\end{fact}
\begin{fact}
Suppose $T(P,R) \ge f, T(Q,R) \ge g,$ then $T(P,R) \ge f \circ g = g(1-f(\alpha))$.
\end{fact} 
 The relationship between $f$-DP and traditional DP has been illustrated from the perspective of hypothesis testing \cite{Dong2022gaussian}. It provides a visual representation of how the choice of parameter $\mu$ in $\mu$-GDP relates to the strength of privacy protection. 
 %Figure  \ref{Fig:fDPComposition} demonstrates that the bound provided by $f$-DP is nearly lossless and significantly outperforms the bounds given by traditional
%$(\epsilon,\delta)$-DP.
%$%%%%%%%%%%%%%%%%%%%%%%%%%%%%%%%%%%%%%%%%%%%%%%%%%%%%%%%%%%%%$
% \begin{figure}
%   \centering
%   % Requires \usepackage{graphicx}
%   \includegraphics[scale = 0.12]{picture/fDPVsDP.jpg}\\
%   \caption{The connection between traditional differential privacy (DP) and $f$-DP can be illustrated as follows. On the left, the function $f_{\epsilon,\delta}$ is a piecewise linear function that is symmetric about the line $y=x$. It has slopes of $-e^{\pm \epsilon}$ and intercepts of $1-\delta$. On the right, the trade-off functions of Gaussian distributions with unit variance and varying means are shown. The line $y=1-x$ represents the absence of privacy leakage.  }\label{Fig:fDPVsDP}
% \end{figure}

The $(\epsilon_i,\delta_i)$-PLDP mechanism can be dominated by the following hypothesis testing problem \cite{Kairouz2015composition}. This forms the foundation for the subsequent analysis.
\begin{lemma}[KOV15]\label{DPtransform}
Let $\mathcal{R}^{(i)}: \mathcal{D} \rightarrow \mathcal{S}$ be an $(\epsilon_i, \delta_i)$-DP local randomizer,
and $x_0,x_1 \in D$, then there exist two quaternary random variables $\tilde{X_0}$ and $\tilde{X_1}$, such that
$\mathcal{R}^{(i)}(x_0)$ and $\mathcal{R}^{(i)}(x_1)$ can be viewed as post-processing of $\tilde{X_0}$ and $\tilde{X_1}$, respectively.
In details,
$$
P(\tilde{X_0}=x) = \left\{
    \begin{array}{ll}
        \delta_i & \mbox{if } x=A, \\
        \frac{(1-\delta_i)e^{\epsilon_i}}{1+e^{\epsilon_i}} & \mbox{if } x = 0, \\
        \frac{1-\delta_i}{1+e^{\epsilon_i}} & \mbox{if } x = 1, \\
        0 & \mbox{if } x = B,
    \end{array}
\right.
$$
and
$$
P(\tilde{X_1}=x) = \left\{
    \begin{array}{ll}
        0 & \mbox{if } x=A, \\
        \frac{1-\delta_i}{1+e^{\epsilon_i}} & \mbox{if } x = 0, \\
        \frac{(1-\delta_i)e^{\epsilon_i}}{1+e^{\epsilon_i}} & \mbox{if } x = 1, \\
        \delta_i & \mbox{if } x = B.
    \end{array}
\right.
$$
\end{lemma}

\subsubsection{Proof of Theorem \ref{thm:post3}}
\begin{proof}
Formally, for each $i \in \{ 2, \cdots, n\}$, let $p_i =  \frac{2(1-\delta_i)}{1+e^{\epsilon_i}}$, we
define random variables $Y_{1,i}^0$, $Y_{1,i}^1$ and $Y_i$ as follows:
\begin{equation}
Y_{1,i}^0 =
\begin{cases}
    \begin{array}{lll}
        0 &  \quad w.p. &\quad e^{\epsilon_i}\frac{p_i}{2}, \\
        1 &  \quad w.p. &\quad \frac{p_i}{2}, \\
        2 &  \quad w.p. &\quad  1-e^{\epsilon_i}\frac{p_i}{2}-\frac{p_i}{2}.
    \end{array}
\end{cases}
\end{equation}
\begin{equation}
Y_{1,i}^1 =
\begin{cases}
    \begin{array}{lll}
        0 &  \quad w.p. &\quad \frac{p_i}{2}, \\
        1 &  \quad w.p. &\quad e^{\epsilon_i}\frac{p_i}{2}, \\
        2 &  \quad w.p. &\quad  1-e^{\epsilon_i}\frac{p_i}{2}-\frac{p_i}{2}.
    \end{array}
\end{cases}
\end{equation}
and
\begin{equation}
Y_i =
\begin{cases}
    \begin{array}{lll}
        0 &  \quad w.p. &\quad \frac{p_i}{2}, \\
        1 &  \quad w.p. &\quad \frac{p_i}{2}, \\
        2 &  \quad w.p. &\quad  1-p_i.
    \end{array}
\end{cases}
\end{equation}
We consider the case in the $t$-th iteration. Given a dataset $X_b$ for $b \in \{0,1\}$, we generate $n$ samples from $\{0,1,2\}$ in the following way. Client number one reports a sample from $Y_{1,i}^b$. Client $i$ ($i=2,\cdots, n)$ each reports an independent sample from $Y_i$. We then shuffle the reports randomly. Let $\rho_b$ denote the resulting distribution over $\{0,1,2\}^n$.
We then count the total number of $0$s and $1$s. Note that a vector containing a permutation of the users responses contains no more information than simply the number of $0$s and $1$s, so we can consider these two representations as equivalent.\\
%Let $y \in \{0,1,2\}^n$ be a permutation of the local reports given by $Y_{1,i}^b$ and $Y_i(i=2,\cdots,n)$. Given the hidden permutation $\pi$, we can generate a sample from
%$\mathcal{A}_S(X_b)$ by sequentially transforming $proc(y_t)$ to obtain the correct mixture components, then sampling from the corresponding mixture component. The difficulty then lies in the fact that conditioned on a particular instantiation $y=v$, the permutation $\pi |_{y=v}$ is not independent of $b$. \\
%The first thing to note is that if $v_t =0$ or $1$, then the corresponding mixture $\mathcal{Q}_1^{0(t)}(z_{1:t-1})$ or $\mathcal{Q}_1^{1(t)}(x_{1:t-1})$,
%is independent of $\pi$. Therefore, in order to do the appropriate post-processing, it suffices to know the permutation $\pi$ restricted to the set of users who sampled $2$,
%$K = \pi(\{i: y_i=2\})$. The set $K$ of users who select $2$ is independent of $b$ since $Y_{1,i}^0$ and $Y_{1,i}^1$ have the same probability of sampling $2$.
We claim that there exists a post-processing function $proc(\cdot)$ such that for $y$ sampled from $\rho_b$, $proc(y)$ is distributed identically to $\mathcal{A}_S(X_b)$.
To see this, let $\pi$ be a randomly and uniformly chosen permutation of $\{1,\cdots,n\}$. For every $i \in \{1,\cdots,n\}$,
given the hidden permutation $\pi$, we can generate a sample from $\mathcal{A}_S(X_b)$ by sequentially transforming $proc(y_t)$ to obtain the correct mixture components, then sampling from the corresponding mixture component. Specially, by Lemma \ref{R2Q},
  \begin{equation}
z_t =
\begin{cases}
    \begin{array}{lll}
        \mathcal{R}^{(t)}(z_{1:t-1},x_1^0) &  \text{if} & y_t = 0; \\
        \mathcal{R}^{(t)}(z_{1:t-1},x_1^1) &  \text{if} & y_t = 1; \\
        \mathcal{R}^{(t)}(z_{1:t-1},x_{\pi(i)}) &  \text{if} & y_t = 2.
    \end{array}
\end{cases}
\end{equation}
By our assumption, this produces a sample $z_t$ from $\mathcal{R}^{(i)}(x_{\pi(i)}).$ It is easy to see that the resulting random variable $(z,y)$ has the property that for
input $b \in \{0,1\}$, its marginal distribution over $\mathcal{S}$ is the same as $\mathcal{A}_S(X_b)$ and its marginal distribution over $\{0,1,2\}^n$ is $\rho_b$.
The difficulty then lies in the fact that conditioned on a particular instantiation $y=v$, the permutation $\pi|_{y=v}$ is not independent of $b$.
Note that if $v_t=0$ or $1$, the corresponding $\mathcal{Q}^{0(t)}_1(z_{1:t-1})$ or $Q_1^{1(t)}(z_{1:t-1}),$ is independent of $\pi$. Therefore, in order to do the appropriate
post-processing, it suffices to know the permutation $\pi$ restricted to the set of users who sampled $2$, $K=\pi(\{i:y_i=2\})$. The set $K$ of users who select $2$ is independent of $b$ since $Y_{1,i}^0$ and $Y_{1,i}^1$ have the same probability of sampling $2$. The probability of being included in $K$ is identical for each $i \in \{2,\cdots,n\},$ and slightly smaller for the first user. Since the sampling of $z$ given $y$ only needs $K$, we can sample from $z|_{(y,K)=(v,J)}$ without knowing $b$. This conditional sampling is exactly the post-processing step that we claimed.

We now analyze the divergence between $\rho_0$ and $\rho_1$, the shuffling step implies that $\rho_0$ and $\rho_1$ are symmetric. This implies that the divergence between $\rho_0$ and $\rho_1$ is equal to the divergence between the distribution between the distribution of the counts of $0'$s and $1'$s.
The decomposition in equation (\ref{DecomposXi}) implies that the divergence between $\mathcal{A}_S(X_0)$ and $\mathcal{A}_S(X_1)$ can be dominated by the divergence of $\rho_0$ and $\rho_1$, where
$\Delta_2 \sim Bern(\delta_{n}), \Delta_0 \sim Bin(1-\Delta_2,\frac{e^{\epsilon_n}}{1+e^{\epsilon_n}})$,
$\Delta_1 = 1 - \Delta_0 - \Delta_2$
 and
 $MultiBern(\theta_1,\cdots, \theta_d)$ represents a $d$-dimensional Bernoulli distribution with $ \sum_{j=1}^d\theta_j = 1$.
\end{proof}

\subsubsection{Proof of Lemma 3}

\begin{proof}
Since  $T(N(\pmb{\mu}_0,\pmb{\Sigma}),N(\pmb{\mu}_1,\pmb{\Sigma}))$ is
$$
\Phi(\Phi^{-1}(1-\alpha)-\sqrt{(\pmb{\mu}_1-\pmb{\mu}_0)'\pmb{\Sigma}^{-1}(\pmb{\mu}_1-\pmb{\mu}_0)}),
$$ 
according to the property of normal distribution, the key is to calculate
$(\pmb{\mu}_1-\pmb{\mu}_0)'\pmb{\Sigma}^{-1} (\pmb{\mu}_1-\pmb{\mu}_0)$.
Let $v_1 = \sum_{i=1}^{n-1} p_i, v_2 = \sum_{i=1}^{n-1}p_i^2,$ then
$$\pmb{\Sigma} =  \left(
\begin{array}{cc}
\frac{v_1}{2}-\frac{v_2}{4} & -\frac{v_2}{4} \\
-\frac{v_2}{4} & \frac{v_1}{2}-\frac{v_2}{4} \\
\end{array}
\right),$$
and
$$\pmb{\Sigma}^{-1} =  \left(
\begin{array}{cc}
\frac{2v_1-v_2}{v_1^2-v_1v_2} & \frac{v_2}{v_1^2-v_1v_2} \\
\frac{v_2}{v_1^2-v_1v_2} & \frac{2v_1-v_2}{v_1^2-v_1v_2} \\
\end{array}
\right).$$
By simple calculation, we can obtain that
$$
(\pmb{\mu}_1-\pmb{\mu}_0)'\pmb{\Sigma}^{-1} (\pmb{\mu}_1-\pmb{\mu}_0)= (-1,1) \pmb{\Sigma}^{-1}(-1,1)' = \frac{4}{\sum_{i=1}^{n-1}p_i}.
$$
Substituting $\mu= \sqrt{\frac{4}{\sum_{i=1}^{n-1}p_i}}$ yields the proof.
\end{proof}
\section*{Acknowledgements}
We would like to thank all the anonymous reviewers for generously dedicating their time and expertise to evaluate our manuscript with insightful comments. 
This work is partially support by JST CREST JPMJCR21M2, JSPS KAKENHI Grant Number JP22H00521, JP22H03595, JP21K19767, JST/NSF Joint Research SICORP JPMJSC2107.
\section*{References}
%\label{sec:reference_examples}
%\bibliographystyle{apa}
\nobibliography*
\bibentry{abadi2016deep}.\\[.2em]
\bibentry{balle2019privacy}.\\[.2em]
\bibentry{berry1941accuracy}.\\[.2em]
\bibentry{Bittau2017prochlo}.\\[.2em]
\bibentry{cheu2019distributed}. \\[.2em]
\bibentry{Dong2022gaussian}.\\[.2em]
%\bibentry{Dwork2006}.\\[.2em]
\bibentry{Dwork2014algorithmic}.\\[.2em]
\bibentry{Erlingsson2019amplification}.\\[.2em]
\bibentry{Esseen1942}.\\[.2em]
\bibentry{Feldman2022hiding}.\\[.2em]
\bibentry{GDDSK21federated}. \\[.2em]
\bibentry{Girgis2021renyi}.\\[.2em]
%\bibentry{GW21Federated}. \\[.2em]
\bibentry{JZT2015}.\\[.2em]
\bibentry{Kairouz2015composition}.\\[.2em]
\bibentry{KS11}.\\[.2em]
\bibentry{lecun1998gradient}. \\[.2em]
\bibentry{liu2021flame}. \\[.2em]
\bibentry{Liu2023} \\[.2em]
%\bibentry{mironov2017renyi} . \\[.2em]
\bibentry{NCW21}. 
%\bibentry{GW21Federated}. \\[.2em]
%\bibentry{WK20federated}.

\nobibliography{aaai22}

\end{document}